\documentclass{sig-alternate}
\usepackage{graphicx}
\usepackage{times}
\usepackage[linesnumbered,ruled,vlined]{algorithm2e}
\usepackage{subfigure}
\usepackage{multirow}
\newtheorem{definition}{Definition}
\newtheorem{theorem}{Theorem}
\newtheorem{lemma}{Lemma}
\newtheorem{corollary}{Corollary}

\def \qed {\hfill \vrule height6pt width 6pt depth 0pt}

\title{Temporal Graph Traversals:\\ Definitions, Algorithms, and Applications}

\author{
{Silu Huang$^{1}$, James Cheng$^{2}$, Huanhuan Wu$^{3}$}
\vspace{1.6mm}\\
\fontsize{10}{10}\selectfont\itshape
\fontsize{10}{10}\selectfont\itshape\rmfamily Department of Computer Science and Engineering, The Chinese University of Hong Kong\\
\fontsize{9}{9}\selectfont\ttfamily\upshape \{$^1$slhuang, $^2$jcheng, $^3$hhwu\}@cse.cuhk.edu.hk
}

\begin{document}
\maketitle

\begin{abstract}


A temporal graph is a graph in which connections between vertices are active at specific times, and such temporal information leads to completely new patterns and knowledge that are not present in a non-temporal graph. In this paper, we study traversal problems in a temporal graph. Graph traversals, such as DFS and BFS, are basic operations for processing and studying a graph. While both DFS and BFS are well-known simple concepts, it is non-trivial to adopt the same notions from a non-temporal graph to a temporal graph. We analyze the difficulties of defining temporal graph traversals and propose new definitions of DFS and BFS for a temporal graph. We investigate the properties of temporal DFS and BFS, and propose efficient algorithms with optimal complexity. In particular, we also study important applications of temporal DFS and BFS. We verify the efficiency and importance of our graph traversal algorithms in real world temporal graphs.

\end{abstract} 
\section{Introduction}  \label{sec:intro}


Graph traversals, such as \emph{depth-first search} (\emph{DFS}) and \emph{breadth-first search} (\emph{BFS}), are the most fundamental graph operations. Both DFS and BFS are not only themselves essential in studying and understanding graphs, but they are also building blocks of numerous more advanced graph algorithms \cite{CormenLRS09mit}. Their importance to graph theory and applications is beyond question.

Surprisingly, however, such basic graph traversal operations as DFS and BFS are not even defined or studied in any depth for an important source of graph data, namely \emph{temporal graphs}. Although both DFS and BFS are simple for a \emph{non-temporal graph}, we shall show that the concepts of DFS and BFS are non-trivial for temporal graphs, which reveal many important properties useful for understanding temporal graphs and lead to new applications.

Temporal graphs are graphs in which vertices and edges are temporal, i.e., they exist or are active at specific time instances. Formally, the temporal graph we study is a graph $G=(V,E)$, where $V$ is the set of vertices and $E$ is the set of edges. Each edge $(u,v) \in E$ is associated with a list of time instances at which $(u,v)$ is active or $u$ is communicating to $v$. A vertex is considered active whenever it is involved in an active edge communication. Figure \ref{fig:tmpG} depicts a Short Message Service (SMS) network modeled as a temporal graph. In the graph, we can see that $a$ sends a message to $b$ at time $1$ and $6$; while $b$ sends a message to $a$ at time $8$.

\begin{figure}[!tbp]
\begin{center}
\includegraphics[width = 1.4in]{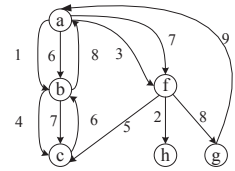}
\vspace{-1mm}
\caption{A temporal graph $G$} \label{fig:tmpG}
\end{center}
\vspace{-3mm}
\end{figure}

\begin{figure}[!tbp]
\begin{center}
\includegraphics[width = 2.4in]{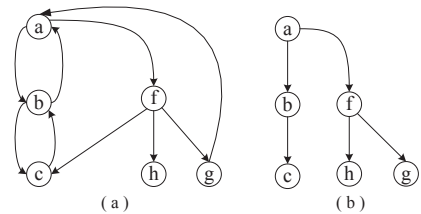}
\vspace{-2mm}
\caption{A non-temporal graph $\hat{G}$ and its DFS/BFS tree} \label{fig:staticG}
\vspace{-4.5mm}
\end{center}
\end{figure}


Although research on graph data has been mainly focused on non-temporal graphs, temporal graphs are in fact ubiquitous in real life. For example, SMS networks, phone call networks, email networks, online social posting networks, stock exchange networks, flight scheduling or travel planning graphs, etc., are all temporal graphs as the objects (i.e., vertices) communicate/connect to each other at different time instances. A long list of different types of temporal graphs is described in \cite{HolmeS11corr}.


Though existing in a wide spectrum of application domains, research on temporal graphs are seriously inadequate, which we believe is mainly due to the common practice of representing a temporal graph as a non-temporal graph for easier data analysis and algorithm design. Figure \ref{fig:staticG}(a) shows the non-temporal graph representation $\hat{G}$ of the temporal graph $G$ in Figure \ref{fig:tmpG}, where the temporal information in $G$ is discarded and multiple edges are combined into one (e.g., the two edges (a,b) at time 1 and 6 are combined into a single edge in $\hat{G}$). Unfortunately, it has been largely overlooked that a non-temporal graph representation actually loses critical information in the temporal graph, which we explain as follows.


Both DFS and BFS are closely related to graph reachability \cite{AgrawalBJ89sigmod,SchaikM11sigmod}, as any path from an ancestor $u$ to a descendant $v$ in the DFS/BFS tree indicates that $u$ can reach $v$. However, a DFS/BFS tree of the non-temporal graph representation does not imply the reachability between the same vertices in the corresponding temporal graph. For example, the DFS and BFS trees, with vertex $a$ as the root, of the non-temporal graph $\hat{G}$ in Figure \ref{fig:staticG}(a) are the same and are given in Figure \ref{fig:staticG}(b). The path $\langle a,f,h\rangle$ in the DFS/BFS tree indicates that $a$ can reach $h$ in $\hat{G}$; however, $a$ cannot reach $h$ in the temporal graph $G$ because in $G$, $a$ reaches $f$ at $3$ and $7$ but $f$ reaches $h$ only before $3$, i.e., $t=2$. In fact, $\langle a,f,h\rangle$ is not a proper path in $G$ since information cannot be transmitted from $a$ to $f$ and then from $f$ to $h$ following a chronological time sequence.


The above discussion not only shows that the non-temporal graph representation is not a good tool for studying temporal graphs, but also motivates the need to define basic graph traversals for temporal graphs. To this end, we conduct the first study on traversal problems in temporal graphs and propose definitions, algorithms, and applications of both DFS and BFS in temporal graphs. We will show that these simple traversal operations on non-temporal graphs become unexpectedly complicated in temporal graphs as the presence of the edges are governed by chronological time sequences.

Note that a temporal graph can be viewed as a sequence of snapshots, where each snapshot is a non-temporal graph in which all edges are active at the same time $t$. A naive approach of temporal graph traversal is to conduct a DFS or BFS in each snapshot. However, such an approach is unrealistic since the number of snapshots in a temporal graph can be very large, e.g., the {\tt wiki} dataset used in our experiment has $134,075,025$ snapshots. Dividing a temporal graph into so many snapshots is not suitable for analysis as it is not easy and efficient to relate the information of one snapshot to that of the next one. 

The main contributions of our work are summarized as follows.

\begin{itemize}
  \item We show that critical temporal information is missing in the non-temporal graph representation of a temporal graph, hence the motivation to study temporal graphs by preserving temporal information.
  \item We show the challenges in defining meaningful DFS and BFS for temporal graphs. We then formally define DFS and BFS in temporal graphs, and design efficient algorithms for their computation.
  \item We believe that basic traversals such as DFS and BFS are the keys to studying temporal graphs, and hence the significance of our work in contributing to future research on temporal graphs that has not been given enough attention so far. As a first step along this direction, we study various graph properties that can be obtained by a DFS or BFS traversal of a temporal graph, and then we identify a set of important applications for both DFS and BFS in temporal graphs.

  \item We conduct extensive experiments on a range of real world temporal graphs. We first evaluate the efficiency of our algorithms. We study the properties of temporal DFS and BFS, and demonstrate their importance by comparing with results obtained from non-temporal graphs. We also show temporal graph traversals are useful in applications.

\end{itemize}

The rest of the paper is organized as follows. In Section \ref{sec:notation} we give the notations. Then, in Sections \ref{sec:dfs} and \ref{sec:bfs} we present the details of temporal DFS and BFS. We discuss applications in Section \ref{sec:app}. We conduct experimental studies in Section \ref{sec:results}. Finally, we discuss related work in Section \ref{sec:related} and give our conclusions in Section \ref{sec:conclusion}.

\section{Notions and Notations}  \label{sec:notation}

We define notions and notations related to temporal graph in this section. We first define two types of closely related edges.

\begin{itemize}
    \item \emph{Temporal edge:} a temporal edge is represented by a triplet, $(u,v,t)$, where $u$ is the start point or start vertex, $v$ is the end point or end vertex, and $t$ is the time when $u$ sends a message to $v$ or when the edge $(u,v)$ is active. we call $u$ the \emph{in-neighbor} of $v$ and $v$ the \emph{out-neighbor} of $u$.
    \item \emph{Non-temporal edge:} a non-temporal edge is simply the conventional edge representation, given by a pair $(u,v)$, where $u$ is the start vertex, and $v$ is the end vertex.
\end{itemize}

Based on the two types of edges, we define temporal graph and non-temporal graph as follows.

\vspace{2mm}

\noindent \emph{Temporal graph:} Let $G=(V,E)$ be a temporal graph, where $V$ is the set of vertices and $E$ is the set of edges in $G$.
\begin{itemize}
\item Each edge $e=(u,v,t) \in E$ is a temporal edge from a vertex $u$ to another vertex $v$ at time $t$. For any two temporal edges $(u,v,t_1)$ and $(u,v,t_2)$, $t_1 \ne t_2$.
\item Each vertex $v \in V$ is active when there is a temporal edge that starts or ends at $v$.
\item $d(u,v)$: the number of temporal edges from $u$ to $v$ in $G$.
\item $E(u,v)$: the set of temporal edges from $u$ to $v$ in $G$, i.e., $E(u,v)= \{ (u,v,t_1),(u,v,t_2),\ldots,(u,v,t_{d(u,v)})\}$.
\item $N_{out}(v)$ or $N_{in}(v)$: the set of out-neighbors or in-neighbors of $v$ in $G$, i.e., $N_{out}(v)=\{u: (v,u,t) \in E\}$ and $N_{in}(v)=\{u: (u,v,t) \in E\}$.
\item $d_{out}(v)$ or $d_{in}(v)$: the temporal out-degree or in-degree of $v$ in $G$, defined as $d_{out}(v)=\sum_{u \in N_{out}(v)} d(v,u)$ and $d_{in}(v)=\sum_{u \in N_{in}(v)} d(u,v)$.
\end{itemize}

Now given a temporal graph $G=(V,E)$, we define the corresponding non-temporal graph $\hat{G}$ of $G$ as follows.

\vspace{2mm}


\noindent \emph{Non-temporal graph:} Given a temporal graph $G=(V,E)$, we construct a non-temporal graph $\hat{G}=(\hat{V},\hat{E})$ from $G$ as follows:
\begin{itemize}
\item $\hat{V}=V$.
\item $\hat{E}=\{(u,v): E(u,v) \subseteq E\}$, i.e., we create a non-temporal edge $(u,v)$ in $\hat{G}$ for every set $E(u,v)$ in $G$.
\item $\hat{N}_{out}(v)$ or $\hat{N}_{in}(v)$: the set of out-neighbors or in-neighbors of $v$ in $\hat{G}$, i.e., $\hat{N}_{out}(v)=\{u: (v,u) \in \hat{E}\}$ and $\hat{N}_{in}(v)=\{u: (u,v) \in \hat{E}\}$.
\item $\hat{d}_{out}(v)$ or $\hat{d}_{in}(v)$: the out-degree or in-degree of $v$ in $\hat{G}$, defined as $\hat{d}_{out}(v)=|\hat{N}_{out}(v)|$ and $\hat{d}_{in}(v)=|\hat{N}_{in}(v)|$.
\end{itemize}


Figures \ref{fig:tmpG} and \ref{fig:staticG}(a) show a temporal graph $G$ and its corresponding non-temporal graph $\hat{G}$. We have $d(a,b)=2$ as $E(a,b)=\{(a,b,1), (a,b,6)\}$, $N_{out}(a)=\hat{N}_{out}(a)=\{b,f\}$ and thus $\hat{d}_{out}(a)$ $=2$, while $d_{out}(a)=d(a,b)+d(a,f)=2+2=4$.

\vspace{2mm}


\noindent \emph{Remarks:} We focus our discussions on directed graphs, but our definitions and algorithms can be trivially applied to undirected graphs. For simplicity, we do not consider self-loops, which can also be easily handled. We also remark that our method can be easily extended to handle temporal edges with a time duration.

\section{Depth-First Search}  \label{sec:dfs}


In this section, we propose two definitions of \emph{depth-first search} (\emph{DFS}) for temporal graphs, and discuss why two definitions are needed. We investigate properties of DFS in temporal graphs and then present efficient algorithms for DFS in temporal graphs.

\subsection{Challenges of DFS}   \label{ssec:dChallenges}


DFS in a non-temporal graph is rather simple, which starts from a chosen source vertex and traverses as deep as possible along each path before backtracking. The DFS constructs a tree rooted at the source vertex. However, in a temporal graph, even such a simple graph traversal problem becomes very complicated due to the presence of temporal information on the edges and the existence of multiple edges between two vertices.

In Section \ref{sec:intro}, we have shown that if we ignore the temporal information, the DFS tree obtained will present incorrect information about the temporal graph. Thus, a DFS in a temporal graph must follow the chronological order carried by the temporal edges.


We can impose a time constraint when traversing a temporal graph. Naturally, the following time constraint should be imposed: when we traverse as deep as possible along a path in the DFS tree, for any two consecutive edges $(u,v,t_1)$ and $(v,w,t_2)$ on any root-to-leave path, we have $t_1 \le t_2$. This constraint is meaningful because if $t_1 > t_2$, then the edge $(u,v,t_1)$ exists after $(v,w,t_2)$ and at time $t_1$ when we traverse $(u,v,t_1)$, the edge $(v,w,t_2)$ no longer exists (as it existed in the past at time $t_2$). For example, in the graph in Figure \ref{fig:tmpG}, $a$ first sent a message to $b$ at \emph{time 1} and then $b$ forwarded it to $c$ at \emph{time 4}, which naturally gives two chronologically ordered edges $(a, b, 1)$ followed by $(b, c, 4)$. On the contrary, if we first have $(a, b, 6)$, then it should not be followed by $(b, c, 4)$ as this order does not give the correct chronological development of events and may lead to a chaotic time sequence especially when the path grows longer.

Imposing the above-mentioned time constraint during temporal graph traversal probably addresses the problem if there is only a single temporal edge going from one vertex to another vertex. However, the existence of multiple temporal edges between two vertices complicates the problem. Consider again the graph in Figure \ref{fig:tmpG}, there are two temporal edges from $a$ to $b$, and the question is how DFS traverses the two edges, are they treated as tree edges or forward edges? The situation is further complicated as there are also multiple temporal edges connecting among their neighbors, leading to a combinatorial effect. Such tricky cases do not occur in a non-temporal graph, and thus careful investigation is needed to define meaningful and useful DFS in temporal graphs.

\subsection{Definitions of DFS}   \label{ssec:dDefinitions}

%
%

We first formally define the time constraint on a temporal graph traversal (including both DFS and BFS) as follows.

\begin{definition} [Time Constraint on Traversal]  \label{de:TC}
Let $u$ be the current vertex during a traversal in a temporal graph, and $\sigma(u)$ be the time when $u$ is visited, i.e., the traversal either starts from $u$ as the source vertex at time $\sigma(u)$, or visits $u$ via a temporal edge $(u',u,\sigma(u))$. Given an edge $e=(u,v,t)$, we traverse $e$ only if $\sigma(u) \le t$.
\end{definition}


The above time constraint was proposed to define temporal paths in \cite{KempeKK02jcss}, and the rationale for setting this time constraint has been explained in Section \ref{ssec:dChallenges}.

To address the problem of multiple temporal edges between two vertices, we allow multiple occurrences of a vertex in a DFS tree, in contrast to a DFS tree in a non-temporal graph in which each vertex appears exactly once. This is reasonable because each vertex is actually active at multiple times when the (multiple) edges are active. We give our first version of DFS in Definition \ref{de:dfs2}.

\begin{definition}[Temporal DFS-v1]  \label{de:dfs2}
Given a temporal graph $G=(V,E)$ and a starting time $t_s$, a DFS in $G$ starting at $t_s$, named as \emph{DFS-v1}, is defined as follows:
\begin{enumerate}
    \item Initialize $\sigma(v)$$=$$\infty$ for all $v$$\in$$V$, and select a source vertex $s$.
    \item Visit $s$ and set $\sigma(s)=t_s$, and go to Step 2(a):
        \begin{enumerate}
            \item \textbf{After visiting a vertex $u$:} Let $E_{u,v}$ be the set of temporal edges going from $u$ to $v$, where each edge $(u,v,t) \in E_{u,v}$ has not been traversed before and $\sigma(u) \le t$.

                If there exists an out-neighbor $v$ of $u$ such that $E_{u,v} \ne \emptyset$, then choose the edge $e=(u,v,t) \in E_{u,v}$, where $t=\min\{t': (u,v,t') \in E_{u,v}\}$, and traverse $e$ and go to Step 2(b).

                If there is no out-neighbor $v$ of $u$ such that $E_{u,v} \ne \emptyset$, then we backtrack to $u$'s predecessor $u'$ (i.e., we have just visited $u$ via the temporal edge $(u',u,t')$) and repeat Step 2(a); or if $u$ is the source vertex, then terminate the DFS. 
            \item \textbf{After traversing a temporal edge $(u,v,t)$:} If $\sigma(v) > t$, we visit the vertex $v$ and set $\sigma(v) = t$, and go to Step 2(a). Else, repeat Step 2(a).
        \end{enumerate}
\end{enumerate}
\end{definition}

Definition \ref{de:dfs2} allows a vertex $v$ to be visited multiple times. The condition ``$\sigma(v) > t$'' in Step 2(b) is necessary to prevent a vertex being both the ancestor and descendant of itself in a DFS tree. For example, if we set ``$\sigma(v) \le t$'', then a DFS of the graph in Figure \ref{fig:tmpG} following the edges $\langle (a, b, 1), (b, c, 4), (c, b, 6), (b, a, 8) \ldots \rangle$ creates a loop $\langle a, b, c, b, a \ldots \rangle$. Thus, setting $\sigma(v)=\infty$ for all $v \ne s$ indicates that initially ``$\sigma(v) > t$'' is satisfied and $v$ can be visited, while setting $\sigma(s)=t_s$ we do not allow $s$ to be visited from any other vertices and hence restrict $s$ to be the root of a DFS tree only.

However, setting the condition ``$\sigma(v) > t$'' alone is not sufficient as there are multiple out-edges we can choose to traverse. Naturally we specify the order of the edges to be traversed to follow the ascending order of the time at which they are active. In addition, since some applications may favor more recent information. Thus, we also allow users to specify a starting time $t_s$ to capture temporal information only at or after $t_s$, while the information before $t_s$ is considered obsolete.

\begin{figure}[!tbp]
\begin{center}
\includegraphics[width = 3in]{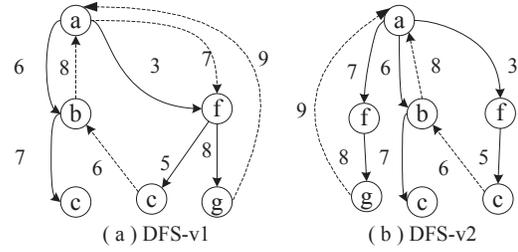}
\caption{DFS in the temporal graph $G$ in Figure \ref{fig:tmpG}} \label{fig:dfs}
\vspace{-4mm}
\end{center}
\end{figure}

Figure \ref{fig:dfs}(a) gives the DFS tree obtained by executing DFS-v1, starting at $t_s=2$, on the graph in Figure \ref{fig:tmpG}. Note that all temporal edges before $t=2$ are neither tree edges nor non-tree edges, as they are considered obsolete. However, it is observed that more recent edges such as $(a, f, 7)$ are not considered as equally as older edges such as $(a, f, 3)$. Thus, if we consider that all temporal edges after a user-specified starting time $t_s$ should receive equal treatments, we will need a new definition of DFS, which we present as follows.

\begin{definition}[Temporal DFS-v2]  \label{de:dfs3}
Given a temporal graph $G=(V,E)$ and a starting time $t_s$, a DFS in $G$ starting at $t_s$, named as \emph{DFS-v2}, is defined as follows:
\begin{enumerate}
    \item Initialize $\sigma(v)$$=$$\infty$ for all $v$$\in$$V$, and select a source vertex $s$.
    \item Visit $s$ and set $\sigma(s)=t_s$, and go to Step 2(a):
        \begin{enumerate}
            \item \textbf{After visiting a vertex $u$:} Let $E_u$ be the set of temporal edges outgoing from $u$, where each edge $(u,v,t) \in E_u$ has not been traversed before and $\sigma(u) \le t$.

                If $E_u \ne \emptyset$, we choose the edge $e=(u,v,t) \in E_u$, where $t=\max\{t': (u,v',t') \in E_u\}$, and traverse $e$ and go to Step 2(b).

                Else (i.e., $E_u = \emptyset$), we backtrack to $u$'s predecessor $u'$ (i.e., we have just visited $u$ via the temporal edge $(u',u,t')$) and repeat Step 2(a); or if $u$ is the source vertex, then terminate the DFS. 
            \item \textbf{After traversing a temporal edge $(u,v,t)$:} If $\sigma(v) > t$, we visit the vertex $v$ and set $\sigma(v) = t$, and go to Step 2(a). Else, repeat Step 2(a).
        \end{enumerate}
\end{enumerate}
\end{definition}

The main difference between Definition \ref{de:dfs3} and Definition \ref{de:dfs2} is that among the set of available outgoing temporal edges from a vertex $u$ in Step 2(a), Definition \ref{de:dfs3} chooses the edges to traverse in reverse chronological order. This may look to be counter intuitive, but we will show that this definition of DFS is meaningful, especially in Section \ref{sec:app} we show how it allows us to answer important path and ``distance'' queries.

Figure \ref{fig:dfs}(b) presents the DFS tree obtained by executing DFS-v2, starting at $t_s=2$, on the graph in Figure \ref{fig:tmpG}. We can see the multiple temporal edges between vertices are equally considered in the DFS and presented in the DFS tree.


\subsection{Notions and Properties Related to DFS}   \label{ssec:dProperties}

Now we present a number of notions related to DFS and some good properties of DFS for temporal graphs.

We first formally define tree edges in a DFS as follows.

\begin{definition} [DFS tree and Tree Edges]  \label{de:treeEdge}
In a DFS of a temporal graph $G=(V,E)$, an edge $e=(u,v,t) \in E$ is a \emph{tree edge} if $e$ is traversed in the DFS and $\sigma(v) > t$ when we traverse $e$ (and then we visit $v$ via $e$ and set $\sigma(v)=t$).

The DFS constructs a \emph{DFS tree}, $T$, which is rooted at the source vertex $s$, where the set of vertices in $T$ is the set of vertices visited in the DFS, and the set of edges in $T$ is the set of all tree edges in the DFS. Since some $v \in V$ may be visited multiple times in the DFS, there may be multiple occurrences of $v$ in $T$. 
\end{definition}


Let $E(t_s)=\{e=(u,v,t): e \in E, t \ge t_s\}$ be the number of temporal edges in  a temporal graph $G=(V,E)$ that are active at or after $t_s$. The following lemmas analyze the bound on the size of the DFS tree. Lemma \ref{le:dOnce} will also be used to analyze the complexity of our algorithms in Section \ref{ssec:dAlgorithms}.

\begin{lemma}   \label{le:dOnce}
In a DFS of $G$ (either by Definition \ref{de:dfs2} or \ref{de:dfs3}), only edges in $E(t_s)$ are traversed and each edge in $E(t_s)$ is traversed at most once.
\end{lemma}

\begin{lemma}   \label{le:dTreeSize}
Let $T$ be the DFS tree of $G$. Then, $T$ has at most $|E(t_s)|+1$ vertices and edges.
\end{lemma}

%
%
%
%

The proof of Lemma \ref{le:dOnce} follows directly from Definition \ref{de:dfs2} or \ref{de:dfs3}, while the proof of Lemma \ref{le:dTreeSize} follows directly from Definition \ref{de:treeEdge} and Lemma \ref{le:dOnce}, and the fact that we visit a vertex in the DFS only when we traverse an edge. We next define non-tree edges traversed in a DFS as follows.

\begin{definition} [Non-Tree Edges]  \label{de:nonTreeEdge}
Given a DFS tree $T$ of $G=(V,E)$, an edge $e=(u,v,t) \in E$ is either a tree edge in $T$, or a \emph{non-tree edge} belonging to one of the following four types:
\begin{itemize}
    \item \emph{Forward edge}: $e$ is a \emph{forward edge} if at the time when the DFS traverses $e$, $u$ is already an ancestor of $v$ in $T$.
    \item \emph{Backward edge}: $e$ is a \emph{backward edge} if at the time when the DFS traverses $e$, $u$ is already a descendant of $v$ in $T$.
     \item \emph{Cross edge}: $e$ is a \emph{cross edge} if at the time when the DFS traverses $e$, $u$ is neither an ancestor nor a descendant of $v$ in $T$.
     \item \emph{Non-DFS edge}: $e$ is a \emph{non-DFS edge} if $e$ is not traversed in the DFS.
\end{itemize}
\end{definition}


The following lemmas and corollary present some important properties of DFS in a temporal graph.

\begin{lemma}   \label{le:AD}
Given a DFS tree $T$ of $G=(V,E)$, if a vertex $u$ is an ancestor of another vertex $v$ in $T$ (or $v$ is a descendant of $u$), then $\sigma(u) \le \sigma(v)$. Here $u$ and $v$ refer to a specific occurrence of vertex $u$ and $v$ in $T$, respectively.
\end{lemma}


\begin{proof}
If $u$ is an ancestor of $v$, then there exists a path $\langle u=w_1, w_2, \ldots, w_p=v \rangle$ such that for each edge $e_i=(w_i, w_{i+1},t_i)$ on the path for $1 \le i < p$, we have $\sigma(w_i) \le t_i \le \sigma(w_{i+1})$ by Steps 2(a) and 2(b) of Definitions \ref{de:dfs2} and \ref{de:dfs3} since $e_i$ is a tree edge. The proof follows as $\sigma(w_i) \le t_i \le \sigma(w_{i+1})$, for $1 \le i < p$, implies that $\sigma(u=w_1) \le \sigma(v=w_p)$.
\end{proof}

\begin{lemma}   \label{le:ADA}
Given a DFS tree $T$ of $G=(V,E)$, a vertex $u$ cannot be both an ancestor and a descendant of another vertex $v$ along the same root-to-leaf path.
\end{lemma}
\begin{proof}

If $u$ is both an ancestor and a descendant of $v$ along the same root-to-leaf path, there exists a path $\langle u, \ldots, v, \ldots, w, u \rangle$ such that each edge on the path is a tree edge. Consider the last edge on the path, i.e., $e=(w,u,t)$. At the time when we traverse $e$ from $w$ to $u$, we have $\sigma(u) \le \sigma(v) \le \sigma(w) \le t$ by Steps 2(a) of Definitions \ref{de:dfs2} and \ref{de:dfs3}. Since $e$ is a tree edge, we require $\sigma(u) > t$ when we traverse $e$ (right before we visit $u$ via $e$), which contradicts to  $\sigma(u) \le t$.
\end{proof}

\begin{corollary}   \label{co:5types}
Given a temporal graph $G=(V,E)$, a DFS of $G$ partitions $E$ into five disjoint subsets.
\end{corollary}
\begin{proof}
The proof follows directly from Definitions \ref{de:treeEdge} and \ref{de:nonTreeEdge}, and Lemma \ref{le:ADA}.
\end{proof}

%
%


We illustrate the concepts by an example. The solid lines in Figure \ref{fig:dfs} are all tree edges; $(a,f,7)$  in Figure \ref{fig:dfs}(a) is a forward edge; $(b,a,8)$ is a backward edge and $(c,b,6)$ is a cross edge in Figures \ref{fig:dfs}(a)-(b); while $(a,b,1)$, $(b,c,4)$ and $(f,h,2)$ are non-DFS edges in Figures \ref{fig:dfs}(a)-(b). Some temporal edges cannot be traversed due to time sequential constraint, e.g., $(b,c,4)$ is non-DFS edge even if it is active after $t_s=2$. It is reasonable not to include such temporal edges since we only keep all useful information concerning DFS starting from $s$. Also note that a temporal edge may belong to different categories for different versions of DFS, e.g., $(a,f,7)$ is a forward edge in Figure \ref{fig:dfs}(a) but a tree edge in Figure \ref{fig:dfs}(b).

We next define the notion of cycle in a temporal graph. Similar to the time constraint on temporal graph traversal given in Definition \ref{de:TC}, cycles in a temporal graph also follow sequential time constraint as defined below. For simplicity, in the following discussion on cycles, whenever we mention a vertex $v$ in a DFS tree $T$, we refer to a specific occurrence of $v$ in $T$.

\begin{definition} [Temporal Cycle]  \label{de:cycle}
Given a temporal graph $G$, a \emph{cycle} $C$ in $G$ is given by a sequence of temporal edges $C=\langle (s=w_1, w_2, t_1), (w_2, w_3, t_2), \ldots, (w_c, w_{c+1}=s, t_c) \rangle$, where $s$ is the start and end vertex of $C$, and $t_1 \leq t_2 \leq \dots \leq t_c$.
\end{definition}

Different from cycles in a non-temporal graph, cycles in a temporal graph have a start vertex in order to satisfy the sequential time constraint. Note that we cannot pick another vertex in the cycle without violating the sequential time constraint. For example, if we choose $v_2$ in $C$ to be the start vertex, without considering the temporal information $C$ is still a cycle, but with the temporal information we have $t_2 \leq \dots \leq t_c \ge t_1$.


A temporal cycle, e.g., $\langle (s=w_1, w_2, t_1), (w_2, w_3, t_2), (w_3, w_4=s, t_3) \rangle$, corresponding to $(\langle (s=a, f, 7), (f, g, 8), (g, a=s, 9) \rangle$ in Figure \ref{fig:dfs}(b), indicates that $s$ delivers information at $t_1$ while it gets information feedback at $t_3$ w.r.t. information fusion. In another application such as flight scheduling, the temporal cycle indicates a person leaving $s$ at $t_1$ and returning $s$ at $t_3$, where the difference between $t_1$ and $t_3$ is referred as \emph{round trip time}.


The following lemma is useful for detecting temporal cycles.

\begin{lemma}   \label{le:cycle}
In a DFS of a temporal graph $G=(V,E)$, if a temporal edge $(u,v,t)$ is a backward edge, then $C=\langle (v=w_1, w_2, t_1), (w_2, w_3, t_2), \ldots, (w_{c-1}, w_c, t_{c-1}), (w_c, w_{c+1}, t_c)=(u,v,t) \rangle$ is a cycle in $G$, where $v$ is the start and end vertex of $C$.
\end{lemma}
\begin{proof}
If $(w_c, w_{c+1}, t_c)=(u,v,t)$ is a backward edge, by Definition \ref{de:nonTreeEdge} $u$ is a descendant of $v$, which means that there is a path $\langle (v=w_1, w_2, t_1), (w_2, w_3, t_2), \ldots, (w_{c-1}, w_c=u, t_{c-1}) \rangle$ and each $(w_i, w_{i+1}, t_i)$ is a tree edge for $1 \le i < c$. Thus, by Steps 2(a) and 2(b) of Definitions \ref{de:dfs2} and \ref{de:dfs3}, we have $\sigma(w_i) \le t_i \le \sigma(w_{i+1})$ and hence $t_1 \leq t_2 \leq \dots \le \sigma(w_c)$. Since $(w_c, w_{c+1}, t_c)=(u,v,t)$ is traversed in the DFS, we have $\sigma(w_c) \le t_c$. Thus, $t_1 \leq t_2 \leq \dots \le t_c$ and $C$ is a cycle in $G$.
\end{proof}

The following definition and lemma are related to reachability in a temporal graph.

\begin{definition} [Temporal Graph Reachability]  \label{de:reachability}
Let $G$ be a temporal graph. A vertex $u$ can \emph{reach} another vertex $v$ (or $v$ is \emph{reachable} from $u$) in $G$ if there is a path from $u$ to $v$ in $G$ such that traversing the path starting from $u$ to $v$ follows the time constraint defined in Definition \ref{de:TC}.
\end{definition}

\begin{lemma}   \label{le:dReachability}

Let $V_T$ be the set of distinct vertices (i.e., multiple occurrences of a vertex $v$ are considered as a single $v$) in the DFS tree $T$ of $G$ (by DFS-v1 or DFS-v2), rooted at a source $s$. Let $V_R$ be the set of vertices in $G$ that are reachable from $s$. Then, $V_T = V_R$.
\end{lemma}
\begin{proof}

First, $V_T \subseteq V_R$ because the simple path in $T$ from $s$ to each $v \in V_T$ is a path in $G$ that satisfies the definition of reachability from $s$ to $v$ as given in Definition \ref{de:reachability}. Next, we prove $V_R \subseteq V_T$, i.e., if there exists a path $\langle (s=w_1, w_2, t_1), (w_2, w_3, t_2), \ldots,$ $(w_k, w_{k+1}=v, t_k) \rangle$ in $G$ from $s$ to each $v \in V_R$, then $v \in V_T$. When visiting $s$ in the DFS, the edge $(s=w_1, w_2, t_1)$ must be traversed according to Step 2(a) of DFS-v1 or DFS-v2 since $t_1 \geq t_s$, which implies that $w_2$ must occur in $T$ with $\sigma(w_2) \leq t_1$ ($w_2$ could be visited via another edge $(s=w_1, w_2, t')$ where $t' < t_1$). Then, $(w_2, w_3, t_2)$ must be traversed since $t_2 \geq t_1 \geq \sigma(w_2)$, and thus $w_3$ must occur in $T$ with $\sigma(w_3) \leq t_2$. By recursive analysis we conclude that $w_{k+1}=v$ must occur in $T$ with $\sigma(v) \leq t_k$, i.e., $v \in V_T$. Thus, $V_T = V_R$.
\end{proof}

In Figure \ref{fig:tmpG}, if $t_s=2$, then the set of reachable vertices from $a$ is the whole vertex set except $h$, which is exactly the set of vertices in the DFS trees in Figures \ref{fig:dfs}(a)-(b).

\subsection{Algorithms and Complexity of DFS}   \label{ssec:dAlgorithms}

Before presenting the algorithms for DFS, we first describe the data format for an input temporal graph $G$. Assuming that edges in $G$ are active at time instances $t_1, t_2, \ldots, t_\tau$, where $t_i < t_{i+1}$ for $1 \le i < \tau$. Let $N_{out}(v,t_i)$ be the set of out-neighbors of a vertex $v$ at $t_i$. We assume that each $N_{out}(v,t_{i+1})$ is collected after $N_{out}(v,t_i)$ as time proceeds, and simply concatenated to $N_{out}(v,t_i)$. Thus, for each $v$, the set of temporal out-edges from $v$ in $G$ is stored as $[(N_{out}(v,t_1), t_1), (N_{out}(v,t_2), t_2), \ldots, (N_{out}(v,t_\tau), t_\tau)]$. For example, the out-edges of $a$ in Figure \ref{fig:tmpG}, i.e., $\{(a,b,1),(a,f,3),$ $(a,b,6),(a,f,7)\}$, are stored as $[(b,1),(f,3),(b,6),(f,7)]$, which is ordered chronologically.

In the discussions of all our algorithms, we assume the above-described data format for the input temporal graph.

We now present the algorithms for DFS in a temporal graph. The algorithm for DFS-v1 is in fact rather straightforward following the description of Definition \ref{de:dfs2}. To analyze the complexity and reduce the cost of some operations in the DFS, we first give the following analysis directly based on Definition \ref{de:dfs2}.

If every individual operation in DFS-v1 uses $O(1)$ time, then by Lemma \ref{le:dOnce} we only use $O(|E|+|V|)$ time. However, checking all out-neighbors $v$ of $u$ such that $E_{u,v} \ne \emptyset$ in Step 2(a) takes $O(d_{out}(u))$ on the set of temporal out-edges of $u$, for each time $u$ is visited. Let $n(u)$ be the number of times a vertex $u$ is visited by DFS-v1. Then, DFS-v1 takes $O(|E| + |V| +\sum_{u \in V} (n(u) * d_{out}(u)))=O(|E| * d_{out}(u) + |V|)$ time, since $O(\sum_{u \in V} n(u)) = O(|E|)$ by Lemma \ref{le:dOnce}. We can reduce the time complexity to $O(|E| * \log d_{out}(u) + |V|)$ by using priority queues to select neighbors to be traversed in Step 2(a). In addition, we can do a binary search to choose the edge $e=(u,v,t) \in E_{u,v}$, where $t=\min\{t': (u,v,t') \in E_{u,v}\}$. Since each temporal edge is traversed at most once, the whole binary search costs at most $\sum_{u \in V} \sum_{v \in N_{out}(u)}\sum_{d=1}^{d(u,v)}log (d)=O(|E|*log(d(u,v)))$. Hence the total time complexity is $O(|E| * (\log d_{out}(u)+ \log d(u,v)) + |V|).$


Next, we propose a linear-time algorithm for DFS-v2, which is in fact optimal. Importantly, we find that the same algorithm can also be applied to solve DFS-v1 to achieve the same linear time complexity.

To begin with, we first present the following important lemma.

\begin{lemma}   \label{le:dTraversalOrder}
In a DFS of $G=(V,E)$ by Definition \ref{de:dfs3}, for any vertex $u \in V$ and for any two temporal edges $e_i=(u,v_i,t_i)$ and $e_j=(u,v_j,t_j)$, where $v_i, v_j \in N_{out}(u)$, if $e_i$ is traversed before $e_j$ in the DFS, then we have $t_i \geq t_j$.
\end{lemma}
\begin{proof}


Let $T$ be the DFS tree. Since $u$ may have multiple occurrences in $T$, we have the following two cases when $e_i$ is traversed before $e_j$ in the DFS. If $e_i$ and $e_j$ are traversed when visiting the same occurrence of $u$, then $t_i \geq t_j$ because Definition \ref{de:dfs3} chooses the edges to traverse in reverse chronological order in Step 2(a). Else, let $u_i$ and $u_j$ be two occurrences of $u$, and assume that $e_i$ and $e_j$ are traversed when visiting $u_i$ and $u_j$, respectively. Since $e_i$ is traversed before $e_j$, $u_i$ must occur before $u_j$ in $T$. Thus, $t_j < \sigma(u_i) \le t_i$, because $e_j$ should be traversed when visiting $u_i$ if $t_j \ge \sigma(u_i)$. In both cases, we have $t_i \geq t_j$.
\end{proof}

Let $E(u)$ be the set of all temporal edges going out from $u$. Lemma \ref{le:dTraversalOrder} essentially implies that we can first order $E(u)$ in descending order of the time at which edges in $E(u)$ are active, and then scan the edges in that order to traverse them during an execution of DFS-v2. This descending order of $E(u)$ is simply the reverse order how the set of temporal out-edges from each $u$ in $G$ is stored, as described at the beginning of Section \ref{ssec:dAlgorithms}. Apparently, since each edge is traversed at most once during a DFS by Lemma \ref{le:dOnce} and now we do not need to search the out-neighbors of $u$ in Step 2(a), the time complexity of DFS-v2 is $O(|E| + |V|)$.




Finally, we remark that DFS-v1 can also be processed by scanning $E(u)$ in reverse order as for DFS-v2, and the resultant DFS tree does not violate Definition \ref{de:dfs2} and hence any related properties/notions presented in Section \ref{ssec:dProperties}.

The following theorem states the complexity of DFS in a temporal graph (the proof follows directly from the discussion above).

\begin{theorem}   \label{th:dComplexity}
Given a temporal graph $G=(V,E)$, DFS-v2 (or DFS-v1) in $G$ uses $O(|E| + |V|)$ time and $O(|E|+|V|)$ space.
\end{theorem}


Note that both the time and space complexity given in Theorem \ref{th:dComplexity} are the lower bound because it is easy to give a temporal graph for which $|E|$ edges are traversed and $|V|$ vertices are visited, and we need $O(|E|+|V|)$ space to keep the graph in memory for random vertex/edge access.

\section{Breadth-First Search}  \label{sec:bfs}

In this section, we define \emph{breadth-first search} (\emph{BFS}) for temporal graphs. We also discuss properties of BFS in temporal graphs and present an efficient algorithm for temporal BFS.

\subsection{Challenges of BFS}   \label{ssec:bChallenges}

%
%
%

BFS in a non-temporal graph starts from a chosen source vertex $s$ and visits all $s$'s neighbors, and then from each neighbor $v$ visits the un-visited neighbors of $v$, and so on until all reachable vertices are visited. However, BFS in a temporal graph is much more complicated due to the presence of temporal information.

BFS and DFS in a temporal graph share many similar challenges that are discussed in Section \ref{ssec:dChallenges}. They both need to follow the time constraint stated in Definition \ref{de:TC} and require multiple occurrences of a vertex in the traversal tree in order to retain critical information. In addition, temporal BFS also poses its own challenges.


One issue unique to temporal BFS is that the path that gives the smallest number of hops may not be the path that reaches from the source to the target at the earliest time, i.e., a short path may take longer time to traverse. For example, in Figure \ref{fig:bfs}(a), the shortest path from $a$ to $f$ has only 1 hop but $a$ reaches $f$ at time $t=7$, while a longer path $\langle a, b, f \rangle$ has 2 hops but $a$ reaches $f$ at an earlier time $t=3$ (such information is important for travel planning).

\begin{figure}[!tbp]
\begin{center}
\includegraphics[width = 2.5in]{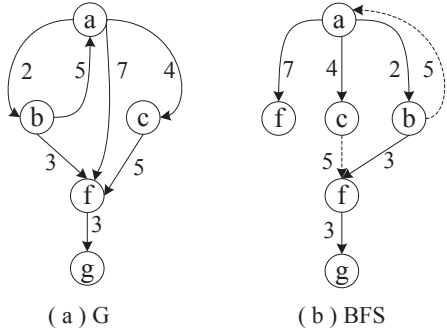}
\caption{BFS in a temporal graph $G$} \label{fig:bfs}
\end{center}
\end{figure}

Figure \ref{fig:bfs}(a) also reveals another problem that has been significantly complicated with the addition of temporal information. When we start from $a$ and finish the first level of BFS (i.e., we have visited $b$, $c$ and $f$), we have $f$ at the second level that can be visited from $b$ and $c$. If we do not visit $f$ again because $f$ has been visited at the first level, then we cannot reach $g$ since $a$ visits $f$ at time $t=7$ and so it cannot go from $f$ to $g$ after $t=7$ since the edge $(f,g)$ is active at $t=3$. If we visit $f$ again, then we can only consider to visit $f$ from $b$ because the time to reach $f$ from $c$ is at time $t=5$, which is also after $t=3$ when the edge $(f,g)$ is active.

\subsection{Definitions of BFS}   \label{ssec:bDefinitions}

To define meaningful and useful BFS for temporal graphs, we need to consider all the issues identified in Section \ref{ssec:bChallenges}.

\begin{definition}[Temporal BFS and BFS Tree]  \label{de:bfs2}
Given a temporal graph $G=(V,E)$ and a starting time $t_s$, a BFS in $G$ starting at $t_s$ is defined as follows:
\begin{enumerate}
    \item Initialize $\sigma(v)=\infty$, $dist(v)=\infty$, and $p(v)=\phi$ for all $v \in V$, where $dist(v)$ denotes the number of hops (or the path distance for un-weighted graphs) from $s$ to $v$ and $p(v)$ denotes the predecessor of $v$ during the BFS; and initialize an empty queue $Q$.
    \item Select a source vertex $s$, set $\sigma(s)=t_s$, $dist(s)=0$, $p(v)=\phi$, and push $(s,dist(s),\sigma(s),p(v))$ into $Q$.
    \item While $Q$ is not empty, do:
        \begin{enumerate}
            \item Pop $(u,dist(u),\sigma(u),p(u))$ from $Q$;
            \item Let $E_{u,v}$ be the set of temporal edges going from $u$ to $v$, where each edge $(u,v,t) \in E_{u,v}$ has not been traversed before and $\sigma(u) \le t$.

                For each out-neighbor $v$ of $u$, where $E_{u,v} \ne \emptyset$, do:
                \begin{enumerate}
                  \item[i.] Let $e=(u,v,t)$ be the edge in $E_{u,v}$ where $t=\min\{t': (u,v,t') \in E_{u,v}\}$.
                  \item[ii.] If $v$ is not in $Q$ (whether $v$ has been visited or not): traverse $e$, and if $\sigma(v) > t$, then visit $v$ and push $(v,dist(v)=dist(u)+1,\sigma(v)=t,p(v)=u)$ into $Q$.
                  \item[iii.] Else:
                  \begin{enumerate}
                    \item[A.] If there exists $(v,dist(v),\sigma(v),p(v))$ in $Q$ such that $dist(v)=dist(u)+1$: traverse $e$, and if $\sigma(v) > t$, then visit $v$ and update $\sigma(v)=t$ and $p(v)=u$ in $Q$.
                    \item[B.] Else (i.e., $dist(v)=dist(u)$): traverse $e$, and if $\sigma(v) > t$, then visit $v$ and push $(v,dist(v)=dist(u)+1,\sigma(v)=t,p(v)=u)$ into $Q$.
                  \end{enumerate}
                \end{enumerate}
        \end{enumerate}
\end{enumerate}

The BFS constructs a \emph{BFS tree}, $T$, which is rooted at $s$, where the set of vertices in $T$ is defined as $(\{s\} \cup \{v: v$ is visited in the BFS and $p(v) \ne \phi\}$, and the set of edges in $T$ is defined as $\{(p(v),v): v$ in $T$ and $p(v) \ne \phi\}$.
\end{definition}

Definition \ref{de:bfs2} addresses the issues identified in Section \ref{ssec:bChallenges}. In addition, we also address the problem of multiple occurrences of a vertex at the same level of the BFS tree; for example, as discussed in Section \ref{ssec:bChallenges}, $f$ at the second level can be visited from $b$ and $c$, and thus two occurrences of $f$ will be created at the second level. For BFS, such multiple occurrences at the same level are of little practical use and may keep much duplicated information.

Figure \ref{fig:bfs}(b) shows the BFS tree of the graph in Figure \ref{fig:bfs}(a). Note that there are still multiple occurrences of a vertex $v \in V$ in the BFS tree, but these occurrences are necessary to keep essential information. For example, $f$ at level 1 in the BFS tree in Figure \ref{fig:bfs}(b) keeps the shortest path from $a$ to $f$, while $f$ at level 2 is necessary to report the shortest path from $a$ to $g$ and also keep the earliest time $a$ can reach $f$ (i.e., at time $t=3$).

Thus, in Definition \ref{de:bfs2}, we keep $\sigma(v)$ updated during the BFS process so that its final value records the earliest time from $s$ to reach this particular occurrence of $v$ in $dist(v)$ hops. It is important to record this earliest time because a later time may miss some paths that start at an earlier time, and hence we may miss the shortest path. For example, $a$ goes to $g$ through $f$ at time $\sigma(f)=3$ instead of $\sigma(f)=7$ via $f$ at level 1.

When traversing an edge $e=(u,v,t)$, if $v$ is in $Q$, then either $v$ was visited at the same level as $u$ but later than $u$ (i.e., $dist(v)=dist(u)$) or $v$ was visited at the current level from another vertex at the same level as $u$ (i.e., $dist(v)=dist(u)+1$). If $v$ is not in $Q$, then either $v$ was visited at a level earlier than $u$ or $v$ was visited at the same level as $u$ but earlier than $u$ or $v$ has never been visited, in either case we need to create another occurrence of $v$ in the BFS tree if $\sigma(v) > t$ since it may be crucial to some paths.

We can verify from the BFS tree that the distance of all the vertices from the source vertex is correctly recorded and the temporal time constraint is also followed along all the paths.




\subsection{Notions and Properties Related to BFS}   \label{ssec:bProperties}

Next we present some important properties and notions related to temporal BFS. 

Let $E(t_s)=\{e=(u,v,t): e \in E, t \ge t_s\}$ be the number of temporal edges in  a temporal graph $G=(V,E)$ that are active at or after $t_s$. The following lemmas analyze the bound on the size of the BFS tree.

\begin{lemma}   \label{le:bOnce}
A BFS of $G$ traverses only edges in $E(t_s)$ and each edge in $E(t_s)$ is traversed at most once.
\end{lemma}

\begin{lemma}   \label{le:bTreeSize}
Let $T$ be the BFS tree of $G$. Then, $T$ has at most $|E(t_s)|+1$ vertices and edges.
\end{lemma}

The proof of Lemma \ref{le:bOnce} follows directly from Definition \ref{de:bfs2}, while the proof of Lemma \ref{le:bTreeSize} follows directly from Definition \ref{de:bfs2} and Lemma \ref{le:bOnce} and the fact that we assign $p(v)$ for a vertex only when we traverse an edge.

The following lemma states that for each $v$ in $G$, there is at most one occurrence of $v$ at each level of the BFS tree. This property will be used to prove some other lemmas.

\begin{lemma}   \label{le:oneOccurrence}
Let $T$ be the BFS tree of $G$, and $O(v)$ be the set of all occurrences of $v$ in $T$. For any $v_1,v_2 \in O(v)$, $dist(v_1) \ne dist(v_2)$.
\end{lemma}
%

The following lemma is related to temporal graph reachability.

\begin{lemma}   \label{le:bReachability}
Let $V_T$ be the set of distinct vertices (i.e., multiple occurrences of a vertex $v$ are considered as a single $v$) in the BFS tree, rooted at $s$, of a temporal graph $G=(V,E)$. Let $V_R$ be the set of vertices in $G$ that are reachable from $s$. Then, $V_T = V_R$.
\end{lemma}

The proof of Lemma \ref{le:oneOccurrence}  follows directly from Definition \ref{de:bfs2}, while the proof of Lemma  \ref{le:bReachability} is similar to that of Lemma \ref{le:dReachability}.

The following definition and lemma are related to the distance between two vertices in a temporal graph.

\begin{definition} [Temporal Graph Distance]  \label{de:distance}
Let $G$ be a temporal graph. The \emph{shortest temporal path distance} in $G$ from a vertex $u$ to another vertex $v$, denoted by $dist(u,v)$, is the minimum number of hops in a path from $u$ to $v$ in $G$ such that traversing the path starting from $u$ to $v$ follows the time constraint defined in Definition \ref{de:TC}.
\end{definition}

\begin{lemma}   \label{le:distance}
Let $T$ be the BFS tree, rooted at $s$, of $G$. Let $v_o$ is the first occurrence of $v$ in $T$. Then, $dist(s,v)= dist(v_o)$.
\end{lemma}
\begin{proof}

BFS is executed level by level, where $dist(v_o)$ is actually the level number of $v_o$ in $T$. Assume that $dist(s,v)=k$ and $P=\langle (s=w_1, w_2, t_1), (w_2, w_3, t_2), \ldots, (w_k, w_{k+1}=v, t_k) \rangle$ is the corresponding path in $G$, we want to prove that $v_o$ is at level $k$ of $T$. First, $w_i$ cannot occur before level $i-1$ with $\sigma(w_i) \leq t_{i}$, otherwise there is another path with a shorter distance than that of $P$ which contradicts to our assumption. Thus, $v$ cannot occur before level $k$ and we need to prove that $v$ occurs at level $k$, i.e., $v_o$ is at level $k$. According to Definition \ref{de:bfs2}, $w_2$ must be at level 1 and $w_3$ cannot be at level 1 with $\sigma(w_2) \leq t_{2}$ (otherwise $dist(s,v)<k$), hence $w_3$ must be at level 2 after traversing $(w_2, w_3, t_2)$. By recursive analysis we can conclude that the first occurrence of $v$ is at level $k$. Thus, we have $dist(s,v)= dist(v_o)$.
\end{proof}


The following lemma specifies the relationship between temporal information and distance information captured in a BFS of a temporal graph.


\begin{lemma}   \label{le:distanceTime}
Let $T$ be the BFS tree, rooted at $s$, of $G$. Let $O(v)$ be the set of all occurrences of $v$ in $T$. Then, for each $v_o \in O(v)$, $\sigma(v_o)$ is the earliest time that $s$ can reach $v$ in $dist(v_o)$ hops in $G$ starting from time $t_s$.
\end{lemma}
\begin{proof}
Suppose $t$ is the earliest time that $s$ can reach $v$ in $i$ hops starting from time $t_s$, via the path $P=\langle (s=w_1, w_2, t_1),$ $(w_2, w_3, t_2),$ $\ldots, (w_k, w_{k+1}=v, t_k=t) \rangle$ in $G$, where $k \leq i$, and $P$ is the shortest path by which $s$ can reach $v$ at time $t$ (starting from $t_s$). Let $v_o$ be the latest occurrence of $v$ in $T$ such that $dist(v_o) \le i$. We want to prove $\sigma(v_o)=t_k=t$. First, $\sigma(v_o) \geq t_k$ (otherwise $t_k$ is not the earliest time that $s$ can reach $v$ in $i$ hops). Next, we prove $\sigma(v_o) \leq t_k$. Similar to the proof of Lemma \ref{le:distance}, we have $w_2$ at level 1 and $w_3$ cannot be at level 1 with $\sigma(w_3) \leq t_2$, hence $w_3$ is at level 2 with $\sigma(w_3) \leq t_2$. A recursive analysis shows that $v$ occurs at level $k$ with $\sigma(v) \leq t_k$, and this is the latest occurrence of $v$ within $k$ hops since there is at most one occurrence of $v$ at level $k$ by Lemma \ref{le:oneOccurrence}. Thus, $\sigma(v_o)$ is the earliest time that $s$ can reach $v$ in $dist(v_o)=k$ hops.
\end{proof}

%
%
%
%
%
%

%
%
%

\subsection{Algorithms and Complexity of BFS}   \label{ssec:bAlgorithms}


The algorithm for the temporal BFS is clear following the description of Definition \ref{de:bfs2}. We first prove a few lemmas as follows, which also analyze the algorithm and its complexity.

\begin{lemma}   \label{le:qSize}
In a BFS of $G$, at most $|E(t_s)|$ records are pushed into the queue $Q$ in Definition \ref{de:bfs2}.
\end{lemma}
\begin{proof}
Since a record is pushed into $Q$ only if an edge is traversed, the proof follows from Lemma \ref{le:bOnce}.
\end{proof}

\begin{lemma}   \label{le:qOccurrence}
There are at most two records involving the same vertex $v$ in $Q$ at any particular time during a BFS of $G$.
\end{lemma}
\begin{proof}
%
%
%
%
%

According to the principle of BFS in Definition \ref{de:bfs2}, at any time the vertices in $Q$ belong to at most two levels, i.e., we have either $dist(v)=i$ or $dist(v)=i+1$ for any $v$ in $Q$, for some positive integer $i$. When considering whether to push a new record of $v$ into $Q$ in Step 3(b)iii, we have two cases that compare with $dist(v)$ of an existing record of $v$ in $Q$: (1) $dist(v) = dist(u)+1$: in which case we only update the existing record in $Q$; or (2)~$dist(v) = dist(u)$, in which case we push a new record for $v$ into $Q$ with $dist(v)=dist(u)+1$. If Case (1) is executed, no new record is created. If Case (2) is executed, let the existing record in $Q$ be $v_{old}$ and the new record be $v_{new}$, and consider that $v$ is visited again, then there exists $dist(v_{new})=dist(u)+1$ now and hence Case (1) will be executed. Thus, no new record of $v$ at level $(dist(u)+1)$ will be pushed into $Q$, and when we process level $(dist(u)+2)$, $v_{old}$ must have been popped from $Q$ according to the principle of BFS.
\end{proof}

Lemmas \ref{le:bOnce} and \ref{le:bTreeSize} show that at most $|E(t_s)|$ edges are traversed in the BFS and at most $|E(t_s)|$ vertices and edges are created. Then, Lemma \ref{le:qSize} shows that at most $|E(t_s)|$ records are pushed into and popped from $Q$, while Lemma \ref{le:qOccurrence} shows that updating a record in $Q$ takes $O(1)$ time since we check at most two records for each $v$ (and there are at most $|E(t_s)|$ updates according to Lemma \ref{le:qSize}).


There is one more operation that we have not considered, that is, computing $E_{u,v}$ for each $v \in N_{out}(u)$ and finding $e$ from $E_{u,v}$ to traverse next, which can be costly if implemented directly. To avoid computing $E_{u,v}$ and search $E_{u,v}$ for $e$, we can use the same sorted $E(u)$ described in Section \ref{ssec:dAlgorithms}, where $E(u)$ is the set of all temporal edges going out from $u$. Since Definition \ref{de:bfs2} does not specify an order by which the out-neighbors of $u$ should be traversed, we can scan the edges in $E(u)$ in the same order as for DFS-v1 to select the next edge to traverse. Since each edge is traversed at most once according to Lemma \ref{le:bOnce}, traversing edges in the sorted $E(u)$ does not violate the definition of BFS. Thus, we have the linear overall complexity as stated by the following theorem.

\begin{theorem}   \label{th:bComplexity}
Given a temporal graph $G=(V,E)$, processing BFS in $G$ uses $O(|E| + |V|)$ time and $O(|E|+|V|)$ space.
\end{theorem}

Similar to DFS, both the time and space complexity given in Theorem \ref{th:bComplexity} are the lower bound for temporal BFS.

\section{Temporal Graph Traversals for Answering Path Queries}  \label{sec:app}

Both DFS and BFS in a non-temporal graph have important applications \cite{CormenLRS09mit}. As the first study (to the best of our knowledge) on temporal graph traversals, we would like to show that our definitions of temporal DFS and BFS also give vital applications.

In Sections \ref{ssec:dProperties} and \ref{ssec:bProperties}, some notions and properties related to temporal DFS and BFS can be readily used in applications such as the detection of cycles, and answering temporal graph reachability queries. Furthermore, these applications are themselves fundamental concepts/tools for studying graphs and therefore each of them has many other applications themselves. For example, temporal graph reachability can be naturally applied to study connected components in a temporal graph.



Due to space limit, we focus on an important set of applications: temporal path queries, such as \emph{foremost paths}, \emph{fastest paths}, and \emph{shortest temporal paths}. We also emphasize that these paths, like shortest paths in a non-temporal graph, can in turn be applied to develop many other useful applications (e.g., temporal graph clustering, temporal centrality computation, etc.).


We first define some common notations used in the definitions of the various types of temporal paths. Given a temporal graph $G$, two vertices $u$ and $v$, and time $t_s$, let $\mathbb{P}$ be the set of all paths in $G$ from $u$ to $v$ and each path in $\mathbb{P}$ implies that $u$ can reach $v$ in $G$ starting at time $t_s$ from $u$. Formally, let $P = \langle (w_1, w_2, t_1), (w_2, w_3, t_2), \ldots,$ $(w_k, w_{k+1}, t_k) \rangle$, then $P$ is in $\mathbb{P}$ if $w_1=u$, $w_{k+1}=v$, $t_1 \ge t_s$, and $w_1$ can reach $w_{k+1}$ in $G$. We define $t_{\it start}(P)=t_1$, $t_{\it end}(P)=t_k$, and $dist(P)=k$.

\subsection{Foremost Paths}   \label{ssec:foremost}

We first define foremost path.


\begin{definition} [Foremost Path]  \label{de:foremost}
A path $P \in \mathbb{P}$ is a \emph{foremost path} if for all path $P' \in \mathbb{P}$, $t_{\it end}(P) \le t_{\it end}(P')$. The problem of \emph{single-source foremost paths} is to find the foremost path from a source vertex $s$, starting from time $t_s$, to every other vertex in $G$.
\end{definition}

Intuitively, a foremost path is the path from $u$, starting from time $t_s$, that reaches $v$ at the earliest possible time. Applications of foremost paths include travel planning for which one wants to know the earliest time to reach a destination if departing at time $t_s$. Figure \ref{fig:tmpG} shows two paths from $a$ to $c$ starting at $t_s=2$, of which $\langle a, f, c \rangle$ is a foremost path while $\langle a, b, c \rangle$ is not.

Next, we show how temporal DFS and BFS can be applied to compute foremost paths.

\begin{theorem}   \label{th:foremost}
Let $T$ be the DFS tree (constructed by DFS-v1) or the BFS tree, rooted at a source vertex $s$, of a temporal graph $G$. Let $O(v)$ be the set of all occurrences of a vertex $v$ in $T$, and let $v_{min} \in O(v)$ be the occurrence of $v$ such that $\sigma(v_{min})=\min\{\sigma(v_o): v_o \in O(v)\}$. For each $v$ in $G$, if $v$ is in $T$, then the path from $s$ to $v_{min}$ in $T$ is a foremost path from $s$ to $v$ in $G$; if $v$ is not in $T$, then the foremost path from $s$ to $v$ does not exist in $G$.
\end{theorem}
\begin{proof}
If $v$ is not in $T$, then the foremost path from $s$ to $v$ cannot exist because Lemmas \ref{le:dReachability} and \ref{le:bReachability} indicate that $s$ cannot reach $v$. 

If $v$ is in the BFS tree $T$, then the proof follows directly from Lemma \ref{le:distanceTime} (by considering the whole tree $T$).

If $v$ is in the DFS tree $T$, then according to Definition \ref{de:dfs2}, after traversing a temporal edge $(u,v,t)$, the DFS visits $v$ only if $\sigma(v)$ can be made smaller (by setting $\sigma(v) = t$); thus, for each occurrence $v_o$ of $v$ in $T$, we always keep the earliest time of $v_o$ that can be reached from $s$. Therefore, the path from $s$ to $v_{min}$ in $T$ is a foremost path.
\end{proof}



We show an example of computing foremost paths using the DFS tree shown in Figure \ref{fig:dfs}(a). When starting from $a$ with $t_s=2$, the earliest time to visit $c$ is $t=5$ by the path $\langle a, f, c \rangle$. 

In addition, when $T$ is the BFS tree, Lemma \ref{le:distanceTime} and Theorem \ref{th:foremost} also imply that the foremost path from $s$ to $v$ with the smallest $dist(v)$ in $T$ is the shortest one among all foremost paths from $s$ to $v$.


\subsection{Fastest Paths}   \label{ssec:fastest}

We define fastest path as follows.


\begin{definition} [Fastest Path]  \label{de:fastest}
A path $P \in \mathbb{P}$ is a \emph{fastest path} if for all path $P' \in \mathbb{P}$, $t_{\it end}(P)-t_{\it start}(P) \le t_{\it end}(P') - t_{\it start}(P')$. The problem of \emph{single-source fastest paths} is to find the fastest path from a source vertex $s$, starting from $t_s$, to every other vertex in $G$.
\end{definition}

Intuitively, a fastest path is the path by which $u$ can reach $v$ using the minimum amount of time among all paths in $\mathbb{P}$. Note that fastest path is different from shortest path (e.g., the shortest route may not be the fastest route to travel), and thus fastest path can be more useful than shortest path in traffic navigation.


A fastest path can be found by finding the foremost path starting at every time instance after $t_s$, but we show by Theorem \ref{th:fastest} that fastest paths can be computed with the same linear time complexity as computing foremost paths.

\begin{definition} [Active Interval]  \label{de:interval}
Let $T$ be the DFS tree constructed by DFS-v2, rooted at a source vertex $s$, of a temporal graph $G$. Given the simple path from $s$ to a vertex $v$ in $T$, i.e., $P=\langle (s=w_1, w_2, t_1), \ldots, (w_k, w_{k+1}=v, t_k) \rangle$ in $T$, the \emph{active interval} of $P$ is given by $[t_{\it start}(P), t_{\it end}(P)]=[t_1, t_k]$.
\end{definition}

Since each $v$ in $T$ represents a path from $s$ to $v$, we simply use $[t_{\it start}(v), t_{\it end}(v)]$ to represent $[t_{\it start}(P), t_{\it end}(P)]$. 

\begin{theorem}   \label{th:fastest}
Let $T$ be the DFS tree constructed by DFS-v2, rooted at a source vertex $s$ starting at time $t_s$, of a temporal graph $G$. Let $O(v)$ be the set of all occurrences of a vertex $v$ in $T$.


Given a time interval $[t_x, t_y]$, where $t_x \ge t_s$, and for each $v$ in $G$, if $v$ is in $T$, let $v_{min} \in O(v)$ be the occurrence of $v$ such that $t_{\it start}(v_{min}) \ge t_x$, $t_{\it end}(v_{min}) \le t_y$, and $t_{\it end}(v_{min})-t_{\it start}(v_{min})=\min\{t_{\it end}(v_o)-t_{\it start}(v_o): v_o \in O(v), t_{\it start}(v_o) \ge t_x, t_{\it end}(v_o) \le t_y\}$, then the path from $s$ to $v_{min}$ in $T$ is a fastest path from $s$ to $v$ in $G$ within the interval $[t_x, t_y]$. If such a $v_{min}$ does not exist, then the foremost path from $s$ to $v$ does not exist in $G$ within $[t_x, t_y]$.
\end{theorem}
\begin{proof}

Let $O(v) = \{v_1, \ldots, v_k\}$ where $v_i$ is the $i$-th occurrence of $v$ in $T$ for $1 \leq i \leq k=|O(v)|$. First, we observe that for $1 \leq i < k$, $t_{\it start}(v_i) \ge t_{\it start}(v_{i+1})$ since edges are traversed in reverse chronological order in Step 2(a) of DFS-v2, while $t_{\it end}(v_i) > t_{\it end}(v_{i+1})$ since $v$ will be visited again only if $\sigma(v)$ can be made smaller in Step 2(b). We assume $t_{\it start}(v_i) > t_{\it start}(v_{i+1})$ since if $t_{\it start}(v_i) = t_{\it start}(v_{i+1})$, we have $t_{\it end}(v_{i+1}) - t_{\it start}(v_{i+1})$ $< t_{\it end}(v_i) - t_{\it start}(v_i)$ and hence we can simply remove $v_i$ from $O(v)$. We want to prove: for each $t \ge t_x$, if the earliest time (starting from $s$ at time $t$) that $s$ can reach $v$ in $G$ is $t'$ and $t' \le t_y$, then there exists an $i$ such that $t_{\it start}(v_{i+1}) < t \le t_{\it start}(v_i)$ and $t_{\it end}(v_i)=t'$, where $1 \leq i \leq k$ and assume $t_{\it start}(v_{k+1})=0$.

Let $t(v_i)$ be the earliest time (starting from $s$ at time $t_{\it start}(v_i)$) that $s$ can reach $v$ in $G$. We want to prove $t(v_i) = t_{\it end}(v_i)$. First, $t_{\it end}(v_i) \geq t(v_i)$ (otherwise $t(v_i)$ is not the earliest time). Next, let $P$ be the foremost path starting at time $t_{\it start}(v_i)$, we can prove that DFS-v2 must traverse all temporal edges on $P$ similar to the proof of Lemma \ref{le:dReachability}, by which we have $t(v_i)=\sigma(v_i) \le t_{\it end}(v_i)$. Thus, $t(v_i)=t_{\it end}(v_i)$. Since $t_{\it start}(v_i)$ is the earliest time that we can start the path from $s$ after time $t$, we have $t(v_i)=t_{\it end}(v_i)=t'$.

By taking $v_{min} \in O(v)$ such that $t_{\it end}(v_{min})-t_{\it start}(v_{min})= \min\{t_{\it end}(v_i)-t_{\it end}(v_i): v_i \in O(v)\}$, the corresponding foremost path starting at time $t_{\it start}(v_{min})$ is a fastest path from $s$ to $v$ within $[t_x, t_y]$.
\end{proof}

Theorem \ref{th:fastest} in fact further generalizes Definition \ref{de:fastest}, which can not only process the special case of Definition \ref{de:fastest} (i.e., $[t_x=t_s, t_y=\infty]$), but can also process single-source foremost paths starting at any $t \ge t_s$ by looking up $v_{min} \in O(v)$ such that $t_{\it start}(v_{min}) \ge t$ and $t_{\it end}(v_{min})=\min\{t_{\it end}(v_o): v_o \in O(v), t_{\it start}(v_o) \ge t\}$.

We illustrate how we compute fastest paths using the DFS tree shown in Figure \ref{fig:dfs}(b). For example, the active intervals for $c$ are $[6,7]$ and $[3,5]$ when starting from $a$ with $t_s=2$. If we set $[t_x=t_s, t_y=\infty]$), then the fastest path should be $P=\langle (a, b, 6), (b, c, 7) \rangle$, which takes on $7-6=1$ time unit to go from $a$ to $c$.  Note that with the same setting, the fastest path from $a$ to $c$ is different from the foremost path in this case, where the foremost path is $P=\langle (a, f, 3), (f, c, 5) \rangle$ which takes $5-3=2$ time units to go from $a$ to $c$.

\subsection{Shortest Temporal Paths}   \label{ssec:shortest}

Finally, we define shortest temporal path.


\begin{definition} [Shortest Temporal Path]  \label{de:shortest}
A path $P \in \mathbb{P}$ is a \emph{shortest temporal path} if for all path $P' \in \mathbb{P}$, $dist(P) \le dist(P')$. The problem of \emph{single-source shortest temporal paths} is to find the shortest temporal path from a source vertex $s$, starting from $t_s$, to every other vertex in $G$.
\end{definition}

The main difference between shortest temporal path and shortest path in a non-temporal graph is that the former follows the time constraint given in Definition \ref{de:TC}. Unlike shortest path in a non-temporal graph, a subpath of a shortest temporal path may not be shortest. For example, in Figure \ref{fig:bfs}(a), the shortest path from $a$ to $g$ is $P=\langle (a, b, 2), (b, f, 3), (f,g,3) \rangle$; however, the subpath $\langle (a, b, 2), (b, f, 3)\rangle$ is not the shortest path from $a$ to $f$, since  $\langle (a, f, 7) \rangle$ uses only 1 hop and is shorter. Thus, algorithms for shortest paths cannot be easily applied to computing temporal shortest paths. The following theorem shows how we compute temporal shortest paths, which can also be easily verified by checking the BFS tree shown in Figure \ref{fig:bfs}(b).





\begin{theorem}   \label{th:distance}
Let $T$ be the BFS tree, rooted at a source vertex $s$, of a temporal graph $G$. Let $v_{o}$ be the first occurrence of $v$ in $T$. For each $v$ in $G$, if $v$ is in $T$, then the path from $s$ to $v_{o}$ in $T$ is a shortest temporal path from $s$ to $v$ in $G$; otherwise, the shortest temporal path from $s$ to $v$ does not exist in $G$.
\end{theorem}
\begin{proof}

First, Lemma \ref{le:bReachability} indicates that for any $v \in V$, if $s$ can reach $v$, then $v$ is in $T$ and by Lemma \ref{le:distance} the path from $s$ to $v_{o}$ gives the shortest temporal path distance. If $v$ is not in $T$, then $s$ cannot reach $v$.
\end{proof}

\subsection{Complexity of Temporal Path Processing}   \label{ssec:pathComplexity}

It is easy to see that computing single-source foremost, fastest, and shortest temporal paths all take only $O(|E|+|V|)$ time as the time for DFS and BFS in a temporal graph, because we only need to traverse the DFS or BFS tree once to obtain all necessary information to compute the paths.

We remark that the above three types of paths were introduced in \cite{XuanFJ03ijfcs}. Although there are many studies related to the applications of temporal paths \cite{CasteigtsFMS11ipps,Holme05phyrev,KempeKK02jcss,KossinetsKW08kdd,Kostakos09physica,NicosiaTMRML11corr,PanS11physrev,SantoroQFCA11corr,TangMML09sigcomm,TangMML10ccr,TangSMML10phyrev} (see details in Section \ref{sec:related}), algorithms for computing these paths were only studied in \cite{XuanFJ03ijfcs}. Let $d_{\it max}=\max\{d(u,v): u,v \in V\}$. The algorithms in \cite{XuanFJ03ijfcs} take  $O(|E|(\log d_{\it max} + \log |V|))$, \ \ $O(|E|^2|V| d_{\it max})$, \ \ and  \ \  $O(|E||V| \log d_{\it max})$ time, respectively, for computing single-source foremost paths, fastest paths, and shortest temporal paths. Thus, their complexity is significantly higher than our linear time complexity. Our experiments in Section \ref{result:path} also verify the huge difference in the performance of our algorithms and theirs. Such results demonstrate the need to define useful temporal graph traversals. As a fact, the algorithms in \cite{XuanFJ03ijfcs} for computing foremost paths and fastest paths are based on Dijkstra's strategy, which is not suitable since these paths are critically determined by the time carried by the last edge on the path. Their algorithm computing shortest paths is more like an enumeration strategy, which is a somewhat straightforward solution for the problem.

\section{Experimental Evaluation}   \label{sec:results}

%

We implemented our temporal graph traversal algorithms (DFS and BFS) in C++ (all source codes, and relevant data for verifying our results, will be made available online). All our experiments were ran on a computer with an Intel 3.3GHz CPU and 16GB RAM, running Ubuntu 12.04 Linux OS.

The goals of our experiments are twofold. First, we study the performance of DFS and BFS in temporal graphs, and stress on the importance for studying temporal graph traversals by comparing with results obtained from the corresponding non-temporal graphs. Second, we study the applications of DFS and BFS to compute different types of temporal paths, and verify our optimal time complexity by comparing with the path algorithms in \cite{XuanFJ03ijfcs} (also coded in C++ and compiled in the same way as our codes).



\subsection{Datasets}  \label{results:real}

We used 7 real temporal graphs from the Koblenz Large Network Collection ({\tt http://konect.uni-koblenz.de/}). We selected the \emph{largest} graphs from each of the following categories: ${\tt arxiv}$-${\tt HepPh}$ (${\tt arxiv}$) from the arxiv networks, in which each temporal edge indicates two authors having a common publication at a specified time instance; ${\tt edit}$-${\tt enwiki}$ (${\tt wiki}$) from the edit networks of Wikipedia,  in which each temporal edge indicates one person editing a Wikipedia page at a specified time; ${\tt enron}$ from the email networks, in which each temporal edge indicates one user sending an email to another user at a specified time; ${\tt facebook}$-${\tt wosn}$-${\tt wall}$ (${\tt fb}$) from the Message Posting Networks, in which each temporal edge indicates one user writing a post on another user's facebook post wall at a specified time; ${\tt dblp}$-${\tt coauthor}$ (${\tt dblp}$) from the DBLP coauthor network, in which each temporal edge indicates two persons coauthoring in a paper at a specific time; ${\tt youtube}$-${\tt u}$-${\tt growth}$ (${\tt youtube}$) from the social media networks, in which each temporal edge indicates a link from one user to another user observed at a specified time; ${\tt delicious}$-${\tt ti}$ (${\tt deli}$) from the bookmark networks, in which each temporal edge indicates one user adding a bookmark at a specified time instance.


Table \ref{tab:realdata} lists some information about the datasets, including the number of vertices ($|V_{G}|$), the number of temporal edges ($|E_{G}|$), and the number of non-temporal edges ($|E_{\hat{G}}|$). Note that the number of vertices in a temporal graph is the same as that in its corresponding non-temporal graph, i.e., $|V_{G}|=|V_{\hat{G}}|$. We also list the average and maximal temporal degree ($d_{\it avg}(v)$ and $d_{\it max}(v)$), and the average and maximal non-temporal degree ($\hat{d}_{\it avg}(v)$ and $\hat{d}_{\it max}(v)$), as well as the average and maximum number of temporal edges from one vertex to another ($d_{\it avg}(e)$ and $d_{\it max}(e)$). In addition, we list the number of snapshots for each temporal graph $G$, denoted by $|ss_G|$, which is the number of distinct time instances at which the edges of $G$ are active.

%

\begin{table}[!htbp]
\caption{Real datasets ($K=10^3$)} \label{tab:realdata}
\vspace{-3mm}
\begin{center}
\small
\resizebox{\linewidth}{!}{
\begin{tabular}{|l||r|r|r|r|r|r|r|}
  \hline
 & ${\tt arxiv}$ & ${\tt wiki}$ & ${\tt enron}$ & ${\tt fb}$ & ${\tt dblp}$ & ${\tt youtube}$ & ${\tt deli}$ \\
\hline \hline
$|V_{G}|$ & 28K & 21504K & 87K & 47
K & 1103K & 3224K & 33791K  \\
\hline
$|E_{\hat{G}}|$ & 6297K & 122075K & 322K & 274K & 8451K & 9377K & 78223K  \\
\hline
$|E_{G}|$ & 12730K & 266770K & 1148K & 877K & 14704K & 12224K & 301254K  \\
\hline
$\hat{d}_{\it avg}(v)$ & 224.14 & 5.68 & 3.69 & 5.84 & 7.66 & 2.91 & 2.31  \\
\hline
$\hat{d}_{\it max}(v)$ & 4909 & 1916898 & 1566 & 157 & 1189 & 83292 &1396  \\
\hline
$d_{\it avg}(v)$ & 453.14 & 12.41 & 13.15 & 18.68 & 13.33 & 3.79 &8.92 \\
\hline
$d_{\it max}(v)$ & 20451 & 3270682 & 32619 & 1430 & 2219 & 106968 & 100627  \\
\hline
$d_{\it avg}(e)$ &  2.02 & 2.19 & 3.57 & 3.20 & 1.74 & 1.30 &3.85  \\
\hline
$d_{\it max}(e)$ & 580 & 285579 & 3904 & 742 & 409 & 2 & 12243  \\
\hline
$|ss_G|$ & 2337 & 134075025 & 220364 & 867939 & 70 & 203 & 1583  \\
\hline
\end{tabular}
}
\end{center}
\end{table}


From Table \ref{tab:realdata}, we can observe that these 7 datasets have rather different characteristics. The datasets ${\tt wiki}$, ${\tt dblp}$, ${\tt youtube}$ and ${\tt deli}$ have larger graph size in terms of $|V_{G}|$, $|E_{G}|$ and $|E_{\hat{G}|}|$. The ${\tt arxiv}$ dataset has much higher average degree than other datasets, while ${\tt wiki}$ has very high maximum degree. The ${\tt wiki}$ dataset also has very high $d_{\it max}(e)$, while ${\tt enron}$, ${\tt fb}$ and ${\tt deli}$ have relatively large $d_{\it avg}(e)$. On the other hand, compared with ${\tt wiki}$, ${\tt youtube}$ has a very small $d_{\it max}(e)$. In addition, the number of snapshots of the graphs also varies significantly.

By selecting datasets from different application sources and with different characteristics, we can demonstrate the robustness of our algorithms and their suitability in different application domains.

%



\subsection{Results of Temporal DFS and BFS}  \label{results:property}

To evaluate the performance of temporal DFS and BFS, we use two sets of source vertices: 100 randomly generated vertices and 100 highest temporal degree vertices. We measure the number of tree edges (denoted by $|T_\textrm{dfs-v1}|$, $|T_\textrm{dfs-v2}|$, and $|T_\textrm{bfs}|$ for DFS-v1, DFS-v2, and BFS, respectively), the number of traversed temporal edges (denoted by $|E_{\it trv}|$), as well as the number of reachable vertices (denoted by $|V_R|$), averaged over the results obtained from the 100 source vertices. We set the starting time $t_s= 0$, i.e., we process all temporal edges.

We report the results in Tables \ref{tab:DFS12_r} and \ref{tab:DFS12_h}. Note that $|V_R|$ is the same for DFS-v1, DFS-v2, and BFS according to Lemma \ref{le:dReachability} and Lemma \ref{le:bReachability}. Also, $|E_{\it trv}|$ does not change for DFS-v1, DFS-v2, and BFS, although the types of edges traversed have different number. The number of non-tree edges of DFS-v1, DFS-v2, and BFS can be computed as $|E_{\it trv}|-|T_\textrm{dfs-v1}|$, $|E_{\it trv}|-|T_\textrm{dfs-v2}|$, and $|E_{\it trv}|-|T_\textrm{bfs}|$, respectively. Since $|V_R|$ and $|E_{\it trv}|$ do not change for the three traversal methods, we only report them once in the tables.

\begin{table}[!htbp]
\caption{Traversal results of random sources} \label{tab:DFS12_r}
\vspace{-3mm}
\begin{center}
\small
\resizebox{\linewidth}{!}{
\begin{tabular}{|l|l||r|r|r|r|r|r|r|}
\hline
 & ${\tt arxiv}$ & ${\tt wiki}$ & ${\tt enron}$ & ${\tt fb}$ & ${\tt dblp}$ & ${\tt youtube}$ & ${\tt deli}$ \\
\hline \hline
$|T_\textrm{dfs-v1}|$  & 85974 & 2037180 & 24231 & 51300 & 911587 & 191936 & 144  \\
\hline
$|T_\textrm{dfs-v2}|$  & 86786 & 3108536 & 37431 & 89337 & 937239 & 231984 & 144  \\
\hline
$|T_\textrm{bfs}|$  & 29159 & 1062208 & 4760 & 20973 & 577573 & 178794 & 144  \\
\hline
\hline
$|E_{\it trv}|$ & 9860K & 5420K & 63K & 180K & 6532K & 471K & 38K  \\
\hline
$|V_R|$ & 25248 & 786359 & 3646 & 12278 & 458516 & 161460 & 145  \\
\hline
\end{tabular}
}
\end{center}
\vspace{-3mm}
\end{table}

\begin{table}[!htbp]
\caption{Traversal results of high-degree sources} \label{tab:DFS12_h}
\vspace{-3mm}
\begin{center}
\small
\resizebox{\linewidth}{!}{
\begin{tabular}{|l|l||r|r|r|r|r|r|r|}
\hline
 & ${\tt arxiv}$ & ${\tt wiki}$ & ${\tt enron}$ & ${\tt fb}$ & ${\tt dblp}$ & ${\tt youtube}$ & ${\tt deli}$ \\
\hline \hline
$|T_\textrm{dfs-v1}|$  & 143172 & 46667260 & 308093 & 176074 & 1860411 & 2199022 & 1054  \\
\hline
$|T_\textrm{dfs-v2}|$  & 144949 & 73353232 & 492177 & 314101 & 1930604 & 2911707 & 1054  \\
\hline
$|T_\textrm{bfs}|$  & 33188 & 20097467 & 51260 & 54028 & 1059935 & 2016305 & 1054  \\
\hline \hline
$|E_{\it trv}|$ & 11444K & 136500K & 827K & 646K & 127776K & 6364K & 578K  \\
\hline
$|V_R|$ & 26723 & 15253044 & 41844 & 31851 & 802647 & 1822059 & 1055  \\
\hline
\end{tabular}
}
\end{center}
\end{table}

%
%

By comparing $|V_R|$ with $|T|$ ($T$ is either $|T_\textrm{dfs-v1}|$, $|T_\textrm{dfs-v2}|$, or $|T_\textrm{bfs}|$) in Tables \ref{tab:DFS12_r} and \ref{tab:DFS12_h}, we can see that each reachable vertex may have multiple occurrences in the DFS or BFS trees. The results also show that $|T_\textrm{dfs-v2}|$ is considerably larger than $|T_\textrm{dfs-v1}|$ for ${\tt wiki}$, ${\tt enron}$, ${\tt fb}$ and ${\tt youtube}$, but is roughly of the same size for the other datasets. This is reasonable because DFS-v2 can answer interval-based fastest path queries which require more information to be kept than answering foremost path queries using DFS-v1. But we note that $|T_\textrm{dfs-v2}|$ is less than twice of $|T_\textrm{dfs-v1}|$ in all cases.

Compared with $|T_\textrm{bfs}|$, however, both $|T_\textrm{dfs-v1}|$ and $|T_\textrm{dfs-v2}|$ are significantly larger except for ${\tt deli}$ (which is because both DFS and BFS trees for ${\tt deli}$ have only two levels). This result is mainly because BFS imposes a relationship between the distance from the source $s$ and the earlier time to reach a vertex from $s$, i.e., Lemmas \ref{le:distanceTime} and \ref{le:distance}, while DFS only imposes the time constraint on the paths. This result also suggests that for answering foremost path queries (see Theorem \ref{th:foremost}), BFS is a more efficient method than DFS-v1 because of its smaller tree size, which is also verified by the running time shown in Tables \ref{tab:paths_r} and \ref{tab:paths_h}.

Comparing the results starting from different sources, Tables \ref{tab:DFS12_r} and \ref{tab:DFS12_h} show that starting the traversals from the high degree vertices can reach significantly more vertices than from randomly chosen sources. This is understandable since in general there are more alternative paths from high degree vertices to reach other vertices. However, $|V_R|$ for ${\tt arxiv}$ does not differ much in Tables \ref{tab:DFS12_r} and \ref{tab:DFS12_h}, which can be explained by the fact that ${\tt arxiv}$ has high average temporal degree and hence the majority of vertices are reachable from any source vertex with reasonably high degree, and we can verify this by comparing $|V_R|$ with $|V_G|$ of ${\tt arxiv}$ in Table \ref{tab:realdata}.


Finally, the running time of the traversals is almost the same as the time reported in Tables \ref{tab:paths_r} and \ref{tab:paths_h} for computing the temporal paths. We can compare the time with $|T_\textrm{dfs-v1}|$, $|T_\textrm{dfs-v2}|$, and $|T_\textrm{bfs}|$, which shows that the running time is in proportion to the DFS or BFS tree size.

\subsection{Results of Different Starting Time}  \label{results:ts}

We further assess the performance of temporal DFS and BFS w.r.t. different starting time. Since the range of time instances varies greatly in different temporal graphs, we set $t_s$ as a proportion of $|ss_G|$ (see $|ss_G|$ in Table \ref{tab:realdata}): $t_s=0/4|ss_G|=0$ (reported in Section \ref{results:property}), $t_s=1/4|ss_G|$, $t_s=2/4|ss_G|$ and $t_s=3/4|ss_G|$.

Tables \ref{tab:DFS12_r_1}, \ref{tab:DFS12_r_2} and \ref{tab:DFS12_r_3} report the results which start DFS or BFS from randomly chosen sources. The results show that as we start at a greater $t_s$ (i.e., a later time), the tree size of both DFS and BFS decreases, and so is the number of traversed temporal edges $|E_{\it trv}|$ and reachable vertices $|V_R|$. However, the decreasing rate varies for different datasets, which is particularly small for ${\tt dblp}$. We checked and found that the temporal edges in  ${\tt dblp}$ concentrate mostly on recent time instances (i.e., after $3/4|ss_G|$). This result demonstrates the need to set a starting time for traversal in some datasets and the usefulness of examining the temporal edges in different time intervals in general.


\begin{table}[!htbp]
\caption{Traversal results of random sources ($t_s=1/4|ss_G|$)} \label{tab:DFS12_r_1}
\vspace{-3mm}
\begin{center}
\small
\resizebox{\linewidth}{!}{
\begin{tabular}{|l||r|r|r|r|r|r|r|}
\hline
 & ${\tt arxiv}$ & ${\tt wiki}$ & ${\tt enron}$ & ${\tt fb}$ & ${\tt dblp}$ & ${\tt youtube}$ & ${\tt deli}$ \\
\hline \hline
 $|T_\textrm{dfs-v1}|$  & 85212 & 1885457 & 18014 & 36600 & 911587 & 80594 & 144  \\
\hline
 $|T_\textrm{dfs-v2}|$  & 85996 & 2864129 & 27375 & 61916 & 937239 & 99387 & 144  \\
\hline
 $|T_\textrm{bfs}|$  & 29107 & 1032306 & 3848 & 18156 & 577573 & 75322 & 144  \\
\hline
\hline
 $|E_{\it trv}|$ & 9849K & 4861K & 46K & 122K & 6532K & 190K & 38K  \\
\hline
 $|V_R|$ & 25234 & 760612 & 2901 & 10915 & 458516 & 69281 & 145  \\
\hline
\end{tabular}
}
\end{center}
\vspace{-5mm}
\end{table}

\begin{table}[!htbp]
\caption{Traversal results of random sources ($t_s=2/4|ss_G|$)} \label{tab:DFS12_r_2}
\vspace{-3mm}
\begin{center}
\small
\resizebox{\linewidth}{!}{
\begin{tabular}{|l||r|r|r|r|r|r|r|}
\hline
 & ${\tt arxiv}$ & ${\tt wiki}$ & ${\tt enron}$ & ${\tt fb}$ & ${\tt dblp}$ & ${\tt youtube}$ & ${\tt deli}$ \\
\hline \hline
 $|T_\textrm{dfs-v1}|$  & 69687 & 458119 & 9674 & 18671 & 911586 & 45295 & 118  \\
\hline
 $|T_\textrm{dfs-v2}|$  & 70230 & 672954 & 14449 & 30555 & 937239 & 56800 & 118  \\
\hline
 $|T_\textrm{bfs}|$  & 27397 & 300995 & 2338 & 12043 & 577573 & 42434 & 118  \\
\hline
\hline
 $|E_{\it trv}|$ & 9145K & 1046K & 24K & 57K & 6532K & 101K & 25K  \\
\hline
 $|V_R|$ & 23937 & 234424 & 1748 & 7924 & 458516 & 39438 & 119  \\
\hline
\end{tabular}
}
\end{center}
\vspace{-5mm}
\end{table}

\begin{table}[!htbp]
\caption{Traversal results of random sources ($t_s=3/4|ss_G|$)} \label{tab:DFS12_r_3}
\vspace{-3mm}
\begin{center}
\small
\resizebox{\linewidth}{!}{
\begin{tabular}{|l||r|r|r|r|r|r|r|}
\hline
 & ${\tt arxiv}$ & ${\tt wiki}$ & ${\tt enron}$ & ${\tt fb}$ & ${\tt dblp}$ & ${\tt youtube}$ & ${\tt deli}$ \\
\hline \hline
 $|T_\textrm{dfs-v1}|$  & 47835 & 24193 & 1917 & 1366 & 813473 & 26088 & 35  \\
\hline
 $|T_\textrm{dfs-v2}|$  & 48101 & 33744 & 3038 & 2165 & 835203 & 33299 & 35  \\
\hline
 $|T_\textrm{bfs}|$  & 25665 & 22810 & 640 & 1165 & 523374 & 24261 & 35  \\
\hline
\hline
 $|E_{\it trv}|$ & 8373K & 43K & 5K & 4K & 5865K & 59K & 2K  \\
\hline
 $|V_R|$ & 22904 & 20220 & 481 & 974 & 416975 & 22528 & 36  \\
\hline
\end{tabular}
}
\end{center}
\end{table}

The results for starting traversals from the high-degree sources follow a similar trend, though in higher magnitude. Thus, we only report $|E_{\it trv}|$ and $|V_R|$ in Table \ref{tab:DFS12_h_ts} due to space limit.

\begin{table}[!htbp]
\caption{Results of high-degree sources (varying $t_s$)} \label{tab:DFS12_h_ts}
\vspace{-3mm}
\begin{center}
\small
\resizebox{\linewidth}{!}{
\begin{tabular}{|l|l||r|r|r|r|r|r|r|}
\hline
\multicolumn{2}{|c|}{} & ${\tt arxiv}$ & ${\tt wiki}$ & ${\tt enron}$ & ${\tt fb}$ & ${\tt dblp}$ & ${\tt youtube}$ & ${\tt deli}$ \\
\hline \hline
\hline
\multirow{2}{*}{$t_s$$=$$1/4|ss(G)|$} & $|E_{\it trv}|$ & 11416K & 114342K & 580K & 460K & 12776K & 2762K & 515K  \\
\cline{2-9}
& $|V_R|$ & 26685 & 14138968 & 31447 & 30072 & 802647 & 853058 & 998  \\
\hline
\multirow{2}{*}{$t_s$$=$$2/4|ss(G)|$} & $|E_{\it trv}|$ & 11115K & 70381K & 252K & 249K & 12773K & 2061K & 310K  \\
\cline{2-9}
& $|V_R|$ & 26387 & 10881799 & 16338 & 25419 & 802436 & 652770 & 720  \\
\hline
\multirow{2}{*}{$t_s$$=$$3/4|ss(G)|$} & $|E_{\it trv}|$ & 10178K & 24187K & 74K & 30K & 12337K & 1337K & 58K  \\
\cline{2-9}
& $|V_R|$ & 25522 & 5733690 & 6482 & 6746 & 775913 & 437211 & 391  \\
\hline
\end{tabular}
}
\end{center}
\vspace{-2mm}
\end{table}

\subsection{Temporal vs. Non-Temporal}  \label{results:nonTemp}

In this experiment we want to show that the non-temporal graph obtained by discarding the temporal information presents misleading information. Table \ref{tab:static} lists the average number of reachable vertices from the randomly chosen and high-degree sources, respectively. Compared with $|V_R|$ obtained from temporal DFS and BFS as reported in Tables \ref{tab:DFS12_r_1}-\ref{tab:DFS12_h_ts}, $|V_R|$ obtained from the corresponding non-temporal graphs is significantly larger. The result is expected as discarding the temporal information removes the time constraint on the paths. Thus, this result confirms the importance of studying temporal graphs directly.

\begin{table}[!htbp]
\caption{$|V_R|$ in non-temporal graphs} \label{tab:static}
\begin{center}
\small
\resizebox{\linewidth}{!}{
\begin{tabular}{|l||r|r|r|r|r|r|r|}
\hline
  & ${\tt arxiv}$ & ${\tt wiki}$ & ${\tt enron}$ & ${\tt fb}$ & ${\tt dblp}$ & ${\tt youtube}$ & ${\tt deli}$ \\
\hline \hline
Random & 28045 & 2976860 & 9952 & 27639 & 849559 & 236991  & 1583  \\
\hline
High-degree & 28045 & 21050813 & 55277 & 35434 &  965408 & 1983367  & 1582  \\
\hline
\end{tabular}
}
\end{center}
\end{table}


%

\subsection{Performance on Applications}  \label{result:path}

According to Theorems \ref{th:foremost}, \ref{th:fastest} and \ref{th:distance}, temporal graph traversals can answer single-source foremost paths, fastest paths, and shortest paths, respectively. We compare the performance of our graph traversal algorithms with the existing algorithms for computing these three types of paths by Xuan et al. \cite{XuanFJ03ijfcs}, denoted by \textbf{XuanFJ}. We set $t_s=0$ for all paths and $[t_x, t_y]=[0,\infty)$ for fastest path queries. Note that for different $[t_x, t_y]$, we only need to scan the active intervals to obtain the fastest paths.

Tables \ref{tab:paths_r} and \ref{tab:paths_h} report the average running time of answering these path queries, starting from randomly chosen and high-degree sources, respectively. The results show that our method is up to orders of magnitude faster than XuanFJ, especially for computing the fastest paths. Since these paths have many applications in studying various properties of temporal graphs \cite{CasteigtsFQS12paapp,HolmeS11corr,KossinetsKW08kdd,Kostakos09physica,PanS11physrev,SantoroQFCA11corr,TangMML09sigcomm,TangMML10ccr,XuanFJ03ijfcs}, the usefulness of our algorithms is clear and their efficiency is vital for processing large temporal graphs which are becoming increasingly common today.


\begin{table}[!htbp]
\caption{Path query time in seconds  (random sources)} \label{tab:paths_r}
\vspace{-3mm}
\begin{center}
\small
\resizebox{\linewidth}{!}{
\begin{tabular}{|l|l||r|r|r|r|r|r|r|}
\hline
\multicolumn{2}{|c|}{} & ${\tt arxiv}$ & ${\tt wiki}$ & ${\tt enron}$ & ${\tt fb}$ & ${\tt dblp}$ & ${\tt youtube}$ & ${\tt deli}$ \\
\hline \hline
\multirow{2}{*}{Foremost} & $Ours$ & 0.0572 & 0.1830 & 0.0006 & 0.0036 & 0.1857 & 0.0160 & 0.0002  \\
\cline{2-9}
 & XuanFJ & 0.1772 & 2.8920 & 0.0103 & 0.0119 &  0.3164 & 0.2327  & 3.342  \\
\hline
\multirow{2}{*}{Fastest} & $Ours$ & 0.0601 & 0.2233 & 0.0008 & 0.0039 & 0.1872 & 0.0186 & 0.0003  \\
\cline{2-9}
  & XuanFJ & 115.26 & 27.98 & 0.1524 & 0.3187 & 3.8738 & 1.9689  & 17.61   \\
\hline
\multirow{2}{*}{Shortest} & $Ours$ & 0.0239 & 0.1387 & 0.0005 & 0.0021 & 0.1108 & 0.0155 & 0.0001  \\
\cline{2-9}
  & XuanFJ & 1.8346 & 26.0324 & 0.0416 & 0.1240 & 1.9953 & 1.6017 & 9.0814 \\
\cline{2-9}
\hline
\end{tabular}
}
\end{center}
\vspace{-3mm}
\end{table}

\begin{table}[!htbp]
\caption{Path query time in seconds (high-degree sources)} \label{tab:paths_h}
\vspace{-3mm}
\begin{center}
\small
\resizebox{\linewidth}{!}{
\begin{tabular}{|l|l||r|r|r|r|r|r|r|}
\hline
\multicolumn{2}{|c|}{} & ${\tt arxiv}$ & ${\tt wiki}$ & ${\tt enron}$ & ${\tt fb}$ & ${\tt dblp}$ & ${\tt youtube}$ & ${\tt deli}$ \\
\hline \hline
\multirow{2}{*}{Foremost} & $Ours$ & 0.0693 & 4.4225 & 0.0106 & 0.0134 & 0.3433 & 0.2235 & 0.0026  \\
\cline{2-9}
 & XuanFJ & 0.1841 & 21.1088 & 0.0197 & 0.0163 & 0.5969 & 0.8307 & 1.6462 \\
\hline
\multirow{2}{*}{Fastest} & $Ours$ & 0.0710 & 5.6701 & 0.0116 & 0.0143 & 0.3885 & 0.2683 & 0.0023  \\
\cline{2-9}
 & XuanFJ & 3778.44 & 5000+ & 83.2 & 9.36 & 1014.63 & 5000+ & 5000+  \\
\hline
\multirow{2}{*}{Shortest} & $Ours$ & 0.0280 & 3.0686 & 0.0042 & 0.0006 & 0.2065 & 0.1961 & 0.0004  \\
\cline{2-9}
  & XuanFJ & 2.4196 & 144.2384 & 0.1453 & 0.1635 & 4.4822 & 5.2822 & 13.5160  \\
\hline
\end{tabular}
}
\end{center}
\vspace{-3mm}
\end{table}

\section{Related Work}  \label{sec:related}

Most existing works on temporal graphs are related to temporal paths. Temporal paths were first proposed in \cite{KempeKK02jcss} to study the connectivity of a temporal network, for which disjoint temporal paths between any two vertices are computed. Later, three different types of temporal paths, namely foremost, fastest and shortest paths, were introduced in \cite{XuanFJ03ijfcs}, and similar definitions of temporal paths, distance, proximity, and reachability (and their applications) were proposed in a number of works \cite{Holme05phyrev,Kostakos09physica,PanS11physrev,SantoroQFCA11corr}. Among them, only \cite{XuanFJ03ijfcs} gave algorithms and formal complexity analysis for computing these paths. In \cite{CasteigtsFMS11ipps,KossinetsKW08kdd}, temporal paths were applied to study the temporal view or information latency of a vertex about other vertices. In \cite{TangMML10ccr}, temporal paths were used to define a number of temporal metrics such as temporal efficiency and clustering coefficient. Temporal paths were also applied to find temporal connected components in \cite{NicosiaTMRML11corr,TangMML10ccr}. In \cite{TangSMML10phyrev}, small-world behavior was analyzed in temporal networks using temporal paths. In \cite{SantoroQFCA11corr}, betweenness and closeness were defined based on temporal paths but not studied in details. In \cite{PanS11physrev}, empirical studies were conducted to measure correlation between temporal paths and temporal closeness. In \cite{ClementiP10}, the speed of information propagation from one time to another was also related to temporal paths. Other than temporal paths, the computability and complexity of exploring temporal graphs were analyzed in \cite{FlocchiniMS09isaac}, a hierarchy of classes of temporal graphs was defined in \cite{CasteigtsFQS12paapp},  and a survey of many prior proposed concepts of temporal graphs was given in \cite{HolmeS11corr}.

Other related work \cite{DBLP:journals/networks/CaiKW97,DBLP:conf/mdm/ChonAA03,DBLP:conf/edbt/DingYQ08}  studied how to compute time-dependent shortest paths. They focused on graphs in which the cost (travel time) of edges between nodes varies with time but edges are active all the time, which are different from temporal graphs where temporal edges become active and inactive frequently with time. Time-dependent graphs are more suitable in applications such as road networks. We are also aware of the work by
\cite{DBLP:journals/talg/DemetrescuI06,DBLP:journals/corr/abs-1302-5549,DBLP:journals/pvldb/RenLKZC11}, which are related to path and query problems in dynamic graphs or time-evolving graphs subjecting to edge insertions, edge deletions and edge weight updates. Though time-evolving graphs are also time-related, it is significantly different from our temporal graphs. First, only a few snapshots exist in time-evolving graphs while the number of snapshots in a temporal graph can be very large. Second, snapshots of a time-evolving graph often have much overlapping among them, while snapshots in a temporal graph usually have a much lower degree of overlapping.

\section{Conclusions}  \label{sec:conclusion}

We proposed formal definitions for DFS and BFS in temporal graphs, studied various properties and concepts related to temporal DFS and BFS, and proposed efficient algorithms that are vital for answering a number of useful path queries in temporal graphs. Our results on a variety of real world temporal graphs verified the high efficiency of our algorithms, the need for temporal graph traversals as to retain important temporal information, and the usefulness of our work in the application of computing various temporal paths.

\small

\bibliographystyle{abbrv}


\bibliography{ref_tgt}

\end{document}